\documentclass[conference]{IEEEtran}

\IEEEoverridecommandlockouts
\usepackage[utf8]{inputenc}
\usepackage[english]{babel}
\usepackage{cite}
\usepackage{amsmath,amssymb,amsfonts}
\usepackage{algorithmic}
\usepackage{textcomp}
\usepackage{xcolor}
\usepackage{multirow} 
\usepackage{balance}
\usepackage[ruled, lined, linesnumbered, commentsnumbered, noend]{algorithm2e}
\usepackage{booktabs}
\usepackage{threeparttable} 
\usepackage[nolist]{acronym}
\def\BibTeX{{\rm B\kern-.05em{\sc i\kern-.025em b}\kern-.08em
    T\kern-.1667em\lower.7ex\hbox{E}\kern-.125emX}}

\usepackage{graphicx}
\graphicspath{{figures/}}

\usepackage{tikz}
\usepackage{mathtools}
\usepackage{amsthm}
\usepackage[]{nohyperref}
\usepackage{cleveref} 
\usepackage{csquotes} 
\usepackage{subcaption}
\usepackage{newfloat}
\DeclareFloatingEnvironment[fileext=frm,placement={tbp},name=Problem]{problem}
\Crefname{problem}{Problem}{Problems}
\def\BibTeX{{\rm B\kern-.05em{\sc i\kern-.025em b}\kern-.08em
    T\kern-.1667em\lower.7ex\hbox{E}\kern-.125emX}}

\newtheorem{theorem}{Theorem}

\let\union\cup

\let\Union\bigcup

\newcommand{\ERSP}{\ensuremath{E^{\text{RC}}}}
\newcommand{\cp}{\ensuremath{\textnormal{c}}} 
\newcommand{\ports}{\ensuremath{P}} 
\newcommand{\MLU}{\ensuremath{\Theta}} 
\newcommand{\demand}{\ensuremath{t}} 
\renewcommand{\epsilon}{\varepsilon}
\renewcommand{\omega}{\varomega}
\newcommand{\univ}{\ensuremath{U}}
\newcommand{\setcover}{\ensuremath{\mathcal{C}}}
\newcommand{\sets}{\ensuremath{\mathcal{S}}}
\newcommand{\ssep}{\mid}

\newcommand{\N}{\ensuremath{\mathbb{N}}}
\newcommand{\be}{\ensuremath{\coloneqq}}
\DeclarePairedDelimiter\ceil{\lceil}{\rceil}

\newsavebox{\jamBox}
\newlength{\jamWidth}
\newcommand{\jamIfToBig}[2]{%
    \savebox{\jamBox}{#2}%
    \settowidth{\jamWidth}{\usebox{\jamBox}}%
    \ifthenelse{\jamWidth < #1}%
        {\usebox{\jamBox}}%
        {\resizebox{#1}{!}{\usebox{\jamBox}}%
    }%
}
\begin{document}
\renewcommand{\topmargin}{-.7in}
\begin{acronym}
\acro{DAG}{directed acyclic graph}
\acro{DEFO}{Declarative and Expressive Forwarding Optimizer}
\acro{ECMP}{Equal Cost Multipath}
\acro{IGP}{Interior Gateway Protocol}
\acro{IE}{Ingress-Egress}
\acro{ISP}{Internet Service Provider}
\acro{IS-IS}{Intermediate System to Intermediate System}
\acro{LC-MCFS}{\textsc{Linecard-Minimum Multi-Commodity-Flow Subnetwork}}
\acro{LDP}{Label Distribution Protocol}
\acro{LER}{Label Edge Router}
\acro{LP}{Linear Program}
\acro{LSR}{Label Switched Router}
\acro{LSP}{Label Switched Path}
\acro{MCF}{Multicommodity Flow}
\acro{MLU}{Maximum Link Utilization}
\acro{MO}{Midpoint Optimization}
\acro{MPLS}{Multiprotocol Label Switching}
\acro{MSD}{Maximum Segment Depth}
\acro{NDA}{non-disclosure agreement}
\acro{OSPF}{Open Shortest Path First}
\acro{PoP}{Point of Presence}
\acro{RSVP}{Resource Reservation Protocol}
\acro{RL2TLE}{Router-Level 2TLE}
\acro{SDN}{Software-defined Networking }
\acro{SC2SR}{Shortcut 2SR}
\acro{SID}{Segment Identifier}
\acro{SPR}{Shortest Path Routing}
\acro{SR}{Segment Routing}
\acro{SRLS}{Segment Routing Local Search}
\acro{TE}{Traffic Engineering}
\acro{TLE}{Tunnel Limit Extension}
\acro{WAE}{WAN Automation Engine}
\acro{w2TLE}{Weighted 2TLE}
\acro{ILP}{Integer Linear Program}
\acro{RSP}{Route Switch Processor}
\acro{NPU}{Network Processing Unit}
\end{acronym}

\title{Green Traffic Engineering by~Line~Card~Minimization
}

\author{\IEEEauthorblockN{%
Daniel Otten,
Max Ilsen,
Markus Chimani,
Nils Aschenbruck}
\IEEEauthorblockA{%
University of Osnabrück, Institute of Computer Science, Osnabrück, Germany\\
Email: \{daotten, max.ilsen, markus.chimani, aschenbruck\}@uos.de}}

\maketitle

\begin{abstract}
    Green Traffic Engineering encompasses network design and traffic routing strategies that aim at reducing the power consumption of a backbone network.
    We argue that turning off linecards is the most effective approach to reach this goal.
    Thus, we investigate the problem of minimizing the number of active line cards in a network while simultaneously allowing a multi-commodity flow being routed and keeping the maximum link utilization below a certain threshold.
    In addition to proving this problem to be NP-hard, we present an optimal ILP-based algorithm as well as a heuristic based on 2-Segment Routing.
    Lastly, we evaluate both approaches on real-world networks obtained from the Repetita Framework and a globally operating Internet Service Provider.
    The results of this evaluation indicate that our heuristic is not only close to optimal but significantly faster than the optimal algorithm,
    making it viable in practice.
\end{abstract}

\begin{IEEEkeywords}
green traffic engineering, segment routing, energy-aware routing, linecards
\end{IEEEkeywords}

\section{Introduction}

Recently, the influence of increased energy consumption on the environment and the economy has become more and more visible. High energy consumption has a severe impact on the environment: not only due to the increase in carbon emissions, but also because of the exploitation of natural resources. Moreover, energy prices have steadily increased over the last few years.
Hence, there are many incentives to work on the reduction of energy consumption.

In this context, the power consumption of \ac{ISP} networks is of special interest, as they require large amounts of energy. For example, a major German \ac{ISP} recently stated that its energy consumption is as high as the energy-consumption of the entire city of Berlin (cf.~\cite{telekom_enviroment}). The infrastructure of these so-called \emph{backbone networks} is designed to handle the maximum daily amount of traffic. But as shown in \cite{DBLP:conf/lcn/SchullerACHS17}, for nearly eight hours a day, the amount of traffic is below $50$\% of this maximum.
Thus, for a third of the day, the networks are overprovisioned in terms of capacity. Reducing the amount of unneeded capacity by turning the corresponding hardware off is a promising approach to save a huge amount of energy.

Lately, several researchers tackled this problem using various \ac{TE} approaches. However, many of these approaches have shortcomings that make them unsuitable in practice.
For example, some of them use \ac{TE} technologies like \ac{MPLS} with the \ac{RSVP}. These techniques tend to induce an unnecessarily high amount of overhead in the network. Thus, we propose an approach that uses \ac{SR}, a newer light-weight technology.

Moreover, researchers often use old energy models. Namely, they either assume that a whole router can be switched off or that the number of inactive links directly affects the energy consumption of the network. Both assumptions do not hold true. It is not possible to turn a router off completely as a router does not only provide connections within the backbone network but also to customer networks.
Furthermore, it does not hold true that the number of active links directly influences the energy consumption. Only when enough links per router are switched to an inactive state, it is possible to turn a linecard off, leading to a significant reduction in energy consumption of the corresponding router.  We address this issue by examining the problem of linecard minimization. The contributions of this paper are: We show that in general this problem is NP-hard. Then, to calculate a nearly optimal solution, we introduce a 2-\ac{SR} based heuristic. We evaluate this heuristic on different backbone topologies, showing that we can achieve near-optimal results with a significantly reduced computing time.

The remainder of this work is structured as follows: At first, in \Cref{sec:background}, we discuss the power consumption of a backbone network and provide some information about \ac{SR}. In \Cref{sec:relatedWork}, we give an overview of the related work. Thereafter, in \Cref{sec:GreenTeProb}, we present the \ac{LC-MCFS} problem, which we prove to be NP-hard. To still obtain near-optimal solutions to the problem in practice, we introduce an \ac{SR}-based heuristic (\Cref{sec:greenSR}).
Lastly, in \Cref{sec:evaluation}, we evaluate this approach on two datasets: one dataset containing real topologies but artificial traffic, and one obtained by measurements in an \ac{ISP} network.

\section{Background}\label{sec:background}
In this section, we present some background information regarding \ac{SR} and the energy consumption of a backbone network.

\subsection{Segment Routing}
\ac{SR} is a source routing paradigm, whose core idea is to split up the path of a packet into multiple parts, the so-called segments.
There are three different ways to define a segment:
by links (in which case they are called adjacency segments), by network services, and by nodes (routers) of the corresponding network.
Most \ac{SR}-based \ac{TE} focuses on the latter, i.e.\ node segments, as this simplifies optimization as well as practical operation while still offering full traffic steering capabilities.
Thus, for the remainder of this paper, we will also rely entirely on node segments.
These are realized by assigning a unique SSID to every node in the network.
The path of a packet is then determined by a stack of node SSIDs, i.e. the segment list, which is added to the packet at the ingress node.
The packet passes all nodes in the stack from top to bottom, and the path between each pair of consecutive segments is determined by the underlying \ac{IGP}.
In the context of the \ac{IGP} \ac{OSPF}, \ac{SR} can be referred to as the concatenation of shortest paths between node segments.

If the maximum number of segments is bound by a number~$k$, we call the routing paradigm $k$-\ac{SR}.
When $k$ equals the number of nodes (routers) in the network, full traffic control is enabled.
In contrast, if $k$ equals $1$, routing fully depends on the \ac{IGP}.
Adding an arbitrary number of segments to a packet is theoretically possible, but often the number of segments is limited by hardware constraints.
In addition, every segment adds a bit of overhead to the packet.
Thus, it is typically a good idea to keep the number of segments as small as possible.
Fortunately, in previous work it was shown that for most purposes, it is sufficient to rely on just two segments~\cite{bhatia}.

The big advantage of \ac{SR} over other \ac{TE} technology is its significantly smaller overhead.
For example, using the \ac{RSVP} to realize \ac{MPLS} tunnels requires configuration and maintenance on every node of the tunnel (c.f. \cite{rfc3209}).
This causes a lot of overhead, especially when the number of tunnels increases.
Hence, this approach does not scale well with the network size.
By comparison, \ac{SR} is a stateless protocol that requires only a configuration on the first node of an \ac{SR}-path.
All information needed is encoded in the packet itself.
Moreover, there is only one requirement for \ac{SR}:
the underlying \ac{IGP} needs a specific extension.
Other protocols such as \ac{LDP} are not required.
Thus, the overhead introduced in the network is reduced, resulting in an improved scalability with the network size.
Overall, good scalability and traffic control capabilities make \ac{SR} a valuable tool for \ac{TE}.

\subsection{Power Consumption of a Backbone Router}\label{sec:power}

As a backbone network consists of routers connected by links between them, the total power consumption of these network is the total power consumption of its routers.
The power consumption of each router originates from three components: the chassis, the \ac{RSP}, and the linecards. The chassis provides cooling and other basic functions. The \acp{RSP} are X86 servers running Linux systems; they execute the routing protocols and control the linecards. Note that \ac{RSP} is the term used by Cisco, other vendors use a different name for the same concept---e.g., the same hardware built by Arista is named Supervisor Module.
Lastly, linecards provide the physical endpoints of the connections established by the corresponding router. These endpoints are so-called ports. In summary, the power consumption $E(N)$ of a network~$N$ can be written as
\begin{align*}
    E(N) = \sum_{R \in N} E_{R}
    = \sum_{R \in N} \left(\ERSP_R + \sum_{\ell \in R} \left(E_{\ell} + \sum_{p \in \ell} E_{p}\right)\right),
\end{align*}
with $E_{R}$ denoting the power consumption of a router~$R$, $E_{\ell}$ the power consumption of a linecard~$\ell$, and $E_{p}$ the power used to provide services by the port~$p$, including the power used for optics. The energy used by the \ac{RSP} and chassis of router~$R$ is denoted by~$\ERSP_R$.

At an initial glance, switching off routers may seem like a suitable goal for a green \ac{TE} approach.
However, this is usually not possible, especially in \ac{ISP} networks.
The routers take a very long time to reboot, and hardware defects often occur during the reboot as the routers are designed for long uptimes.
Moreover, especially at the edge of the network, routers are used for peering with customer networks. If they are turned off, customers may be disconnected from the network.

However, some linecards like the A99-8X100GE-TR linecard are built with a green design in mind and can be switched off when they process no traffic.
In \cite{power_consumption} it was shown that linecards consume the majority of the power within a router.
Moreover, Cisco provides a power calculator tool~\cite{cisco-power}, which gives an overview of the power consumption of each component used within a router.
With this tool it is easy to confirm that the majority of power is consumed by the linecards---example output is presented in \Cref{tab:Power_consumption_overview1}.
Furthermore, the authors of \cite{Modeling_power_consumption} studied the power consumption of backbone routers considering different load scenarios.  The authors showed that the amount of traffic consumed by modern linecards is nearly independent of the amount of traffic they process. Even when only half of the ports of a linecard are active, the power consumption remains the same. Thus, minimizing the number of active links in the network has no direct effect on the power consumption. As stated before, it is usually not possible to turn a whole router off to save energy.
Therefore, the only suitable target function for a green \ac{TE} approach is to minimize the number of active linecards within the network.
Only if one can switch a linecard to an inactive state, it is possible to save a significant amount of energy. In this paper, we address this topic by developing a \ac{TE} approach which aims to minimize the number of active linecards within a network.
\begin{table}
	\centering
	\caption{Power consumption of an ASR-9912 router based on \cite{cisco-power}}
	\label{tab:Power_consumption_overview1}
	\scriptsize
	\begin{tabular}{l r}
		\toprule
		Component & Power Consumption\\
		\midrule
         A99-8X100GE-TR  & 1100 W \\
         A99-8X100GE-TR  &  1100 W \\
         A99-8X100GE-TR  & 1100 W \\
         \textbf{Linecards Overall} & \textbf{3300 W} \\
         A9K-RSP5-SE  & 275 W\\
         A9K-RSP5-SE  & 275 W\\
         \textbf{RSPs Overall} & \textbf{550 W}\\
         Fans and Chassis & 950 W \\
         \textbf{Total Power Consumption} & \textbf{4800 W}\\
		\bottomrule
	\end{tabular}
\end{table}

\section{Related Work}\label{sec:relatedWork}
Many researchers have developed approaches to minimize the power consumption of different network types. Some researchers like \cite{Energy_Aware_Managment,Energy_Aware_Managment-SPR,Reducing-Power-Consumption} aim to minimize the number of active routers in the network. As discussed before, this is no suitable approach as shutting routers off can possibly disconnect participants from the network.

Most of the existing work like \cite{Disruption_Time-Aware_Green,Dynamic_Link_Sleeping, Load-Dependent-Power} focuses on switching links to an inactive state. This is easier to realize as only the traffic that was routed over the turned off links has to be rerouted. However, as mentioned before, energy-saving is only possible if enough links per router are deactivated to switch a whole linecard to an inactive state. To the best of our knowledge, there is no approach that tackles this issue.

Only the work of \cite{Zhang} uses the concept of linecards. They propose an ILP-based algorithm which reduces the number of active links within the network.  The authors focus on k-shortest path for traffic engineering. This means, they shrink the set of all possible paths to paths that are only a bit longer than the shortest path between two routers. Even on small instances, with 20 nodes, it was not possible to calculate an optimal solution. Thus, we state that this algorithm cannot be used in larger \ac{ISP} backbone networks, which can rach sizes of several hundred nodes. Moreover, the authors assume that one connection is realized by just one linecard. Thus, this approach is minimizing the number of active links in a network.

In the following, we highlight some of the many different algorithmic approaches that have been developed to turn links off. Some authors like \cite{Augmenting-EEE} or \cite{Energy_Aware_Managment} rely on metric tuning for \ac{TE}. The core idea of metric tuning is to change the link metrics in a way that a predefined target is reached. Thus, it relies fully on the \ac{OSPF} \ac{IGP}. Network operators often hesitate to change the link metrics as this can cause a lot of negative unwanted side effects.
Therefore, the metrics are changed as rarely as possible. Especially for a short-term goal like switching some links off, it is not an appropriate method.

Other authors like
\cite{An_SDN_energy_saving_method,Per_Packet_based,Intelligent_Path,GreyWolf,Energy_Consumption_Optimization_for_SoftwareDefined_Networks,Dynamic_Link_Sleeping,green_sr,An_SDN-based_energy-aware_traffic} use a \ac{SDN} controller to route the traffic through the network.
An \ac{SDN} controller is a centralized instance that controls the traffic flow based on the information provided by the routers. To control the traffic, most authors allow arbitrary paths through the network using \ac{MPLS} or similar technologies.
As mentioned before, this can cause a lot of overhead and reduce the network performance.

Therefore, the authors of \cite{green_sr} decided to use a \ac{SR}-based \ac{SDN} approach.  They introduce a simple algorithm to check if a link with low utilization can be turned off. At first, the algorithm tests whether turning this link off splits the network into two disconnected parts. In a second step, a \ac{SR} policy is calculated such that the \ac{MLU} is lower than a predefined threshold.
However, these approaches are strongly based upon a centralized management unit. Such highly automated systems always feature the risk of an accidental misconfiguration, which can have a tremendous influence on the overall network performance. Thus, network operators often prefer a single stable solution that can be tested before applying it to the network.
In this paper, we address the need for an approach, that minimizes the power consumption of a router by minimizing the number of active linecards with a single stable SR configuration, that once it is brought to the network can be used over multiple hours.

\section{Model and Problem Complexity}\label{sec:GreenTeProb}
In this section, we explain our graph-based model of the network and give an \ac{ILP} formulation of the energy saving problem we are tackling herein.
Moreover, we show that this problem is NP-hard.

\subsection{Graph-based Model}\label{sec:graphModel}
We model the network as a directed graph $G=(V,A)$. The set of vertices~$V$ represents active routers in the network, and the set of arcs~$A$ represents its links. Note that we use a multigraph instead of a simple graph to allow for parallel links. As routers establish full duplex connections, it is necessary to model each link as two directed arcs, i.e., one link between router~$u$ and router~$v$ is represented by a set of two distinct arcs~$\{a_{uv},a_{vu}\}$. One such pair of arcs representing a connection is made up of ports on the endpoints. The capacity of this connection is determined by the number and types of the ports that build up the connection. Let $\ports(a)$ denote the set of ports establishing a connection represented by arc~$a$. It holds that
\begin{align*}
    \ports(a_{uv}) = \ports(a_{vu})
\end{align*}
as one port builds a full duplex connection. The capacity of a port is denoted as $c_{p}$. Thus, the capacity of an arc is given by
\begin{align*}
    \cp(a_{uv}) = \cp(a_{vu}) = \sum_{p \in \ports(a_{uv})}\cp_{p}.
\end{align*}
Both arcs $a_{uv}$ and $a_{vu}$ have the same capacity as they are established by the same set of ports. However, their utilization can differ since there can be more traffic routed from $u$ to $v$ than from $v$ to $u$.
\begin{problem*}
\centering

    \begin{flalign}
     & \makebox[0pt][l]{$\displaystyle{}\text{min }   \sum_{v \in V}  \ceil*{ {\frac{\sum_{p \in v} \pi_{p} }{k}} } $} \\ 
     & \text{s.t.}  & \sum_{(x,y) \in V^{2}} f_{xy}(a) &\; \leq \; \theta  		\sum_{p \in \ports(a)} \pi_{p} c_{p} & &\forall a \in A \\
     \setlength{\nulldelimiterspace}{0pt}
     & & \sum_{a_{uv} \in A} f_{xy}(a_{uv}) - \sum_{a_{vu} \in A} f_{xy}(a_{vu}) &\; = \;
            \left\{\begin{IEEEeqnarraybox}[\relax][c]{l's}
                    -\demand_{xy} & if $v=x$\\
                    \demand_{xy} & if $v=y$\\
                0 & else
            \end{IEEEeqnarraybox}\right.
        & & \forall (x,y,v) \in V^3\\
& & \sum_{p \in \ports(a_{uv})} \pi_{p} &\; = \; \sum_{p \in \ports(a_{vu})} \pi_{p} & & \forall a_{uv} \in A \\
& & \pi_{p} &\; \in \lbrace 0 ,1 \rbrace \;  & &\forall p \in \bigcup_{a \in A} \ports(a) 
    \end{flalign}

\medskip
\caption{LC-MCFS formulated as an \ac{ILP}.}
\label{problem:ideal_lp}
\end{problem*}

\subsection{Problem Description and ILP Formulation}
Let us now introduce the \acf{LC-MCFS} problem and the corresponding \ac{ILP}, see \Cref{problem:ideal_lp}.
In this problem, we are given a directed network~$G=(V,A)$ with capacities~$\cp\colon A \to \N$ as described in \Cref{sec:graphModel},
and traffic demands~$\demand: V \times V \to \N$ that need to be routed through $G$.

As shown in \Cref{sec:power}, it is necessary to minimize the number of active linecards within the network.
Here, we have to distinguish between two different scenarios: either the port-to-linecard mapping is fixed, or it is possible to reconfigure the linecards.
Our algorithms deal with the latter scenario; we assume that we are able to assign ports to arbitrary linecards and that the linecards can possibly be turned off during times of low utilization.
In our \ac{ILP}, we model this by introducing a binary variable~$\pi_{p}$ for every port~$p$ to represent whether $p$ is active or not.
Assuming a fixed number of ports~$k$ that a linecard can hold, this yields the following target function:
\begin{align*}
\sum_{v \in V}  \ceil*{\frac{\sum_{p \in v} \pi_{p} }{k}}.
\end{align*}

The first type of constraints in \Cref{problem:ideal_lp} ensures that the total sum of flow routed over an arc~$a$ cannot exceed $\theta$ times its capacity.
This way, we keep the \ac{MLU} below a threshold $\theta$ in order to prevent congestion.
When denoting the amount of $x$-$y$-flow routed over an arc $a$ by $f_{xy}(a)$, the \ac{MLU} $\MLU$ can be expressed as
\begin{align*}
    \MLU = \max_{a \in A }\, \frac{\sum_{(x,y) \in V^{2}} f_{xy}(a)}{ \sum_{p \in \ports(a)} \pi_{p} \cp_{p}} \leq \theta.
\end{align*}

The second and third type of constraints in \Cref{problem:ideal_lp} model the conservation of flow and the establishment of duplex connections, respectively.

\subsection{NP-Hardness}\label{sec:npHardness}
To show that \ac{LC-MCFS} is also of theoretical interest, we first prove its NP-hardness for the restricted case of acyclic graphs and then extend our proof to fit our model of two opposing directed arcs for each link (i.e., full duplex connections).
In the proof, we assume that each arc has exactly two corresponding ports (one per connected router):
this simplifies the objective function to $\sum_{v \in V} \ceil{\frac{\deg_{G'}(v)}{k}}$ where $\deg_{G'}(v)$ denotes the degree of vertex~$v$ in the solution subgraph~$G'=(V,A')$ which contains exactly those arcs that are left active.

However, we will not only discuss the scenario where ports can be assigned to arbitrary linecards but also the one where the port-to-linecard mapping is fixed.
There, we are additionally given an assignment of ports to linecards, i.e., for each vertex~$v$, a partitioning of its incident arcs into $\ceil{\frac{\deg_{G}(v)}{k}}$ sets of size at most $k$.
These sets correspond to linecards that can only be turned off if all of their respective ports are turned off, i.e., the corresponding incident arcs are not part of $A'$.

\begin{figure}[tbp]
\begin{center}
\begin{tikzpicture}[scale=1, transform shape]
        \begin{scope}[every node/.style={circle,minimum size=1pt,inner sep=1pt}]
            \node (a) at (-3,4) {$a$};
            \node (b) at (-1,4) {$b$};
            \node (c) at (1,4) {$c$};
            \node (d) at (3,4) {$d$};
            \node (s1) at (-2,2) {$S^0_1$};
            \node (s2) at (0,2) {$S^0_2$};
            \node (s3) at (2,2) {$S^0_3$};
            \node (s1_1) at (-2,1) {\raisebox{4pt}{$\vdots$}};
            \node (s2_1) at (0,1) {\raisebox{4pt}{$\vdots$}};
            \node (s3_1) at (2,1) {\raisebox{4pt}{$\vdots$}};
            \node (s1_2) at (-2,0) {$S^q_1$};
            \node (s2_2) at (0,0) {$S^q_2$};
            \node (s3_2) at (2,0) {$S^q_3$};
            \node (t) at (0,-1) {$z$};
        \end{scope}

        \begin{scope}[every edge/.style={draw=lightgray!80!black, very thick,->}]
            \path (a) edge (s1);
            \path (b) edge (s1);
            \path (b) edge (s3);
            \path (c) edge (s1);
            \path (c) edge (s2);
            \path (c) edge (s3);
            \path (d) edge (s2);
        \end{scope}

        \begin{scope}[every edge/.style={draw=black, ultra thick,->}]
            \path (s1) edge (s1_1); \path (s1_1) edge (s1_2); \path (s1_2) edge (t);
            \path (s2) edge (s2_1); \path (s2_1) edge (s2_2); \path (s2_2) edge (t);
            \path (s3) edge (s3_1); \path (s3_1) edge (s3_2); \path (s3_2) edge (t);
        \end{scope}
\end{tikzpicture}
\end{center}
    \caption{\ac{LC-MCFS} instance constructed from the Set Cover instance
    $(\univ = \{a,b,c,d\}, \sets=\{\{a,b,c\},~\{c,d\},~\{b,c\}\})$.
    Gray arcs have capacity 1, thick black arcs have capacity $|\univ|$.
    A traffic of 1 needs to be routed from every vertex $v_u$ with $u \in \univ$ to $z$.}
    \label{fig:set_cover_to_lcmcf_reduction}
\end{figure}
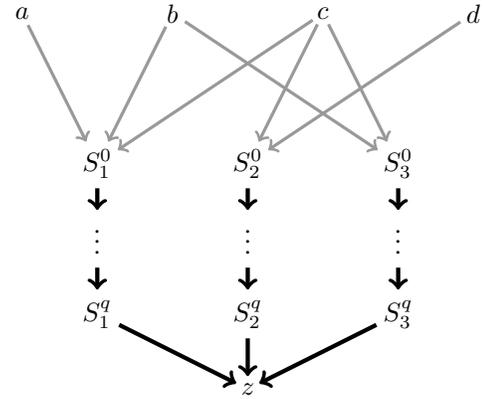
\begin{problem*}
\centering

\setcounter{equation}{0}
    \begin{flalign}
     & \makebox[0pt][l]{$\displaystyle{}\text{min }    \sum_{v \in V}  \ceil* {\frac{\sum_{p \in v} \pi_{p} }{k}}  $} \\
     & \text{s.t.} & \sum_{w \in V \backslash \lbrace u \rbrace}{x^{w}_{uv}}   &\; = \;1 & &\forall (u,v) \in V^{2}  \\
     & & x_{uv}^{w} &\; \geq \; 0 & &\forall  (u,v) \in V^{2}  \\
     & & \sum_{(u,v) \in V^{2}}  \sum_{w \in  V \backslash \lbrace
    u \rbrace} g_{uv}^{w}(a) x^w_{uv}  &\; \leq \; \theta \sum_{p \in \ports(a)} \pi_{p} c_{p} & &\forall a \in A \\
     & & \sum_{p \in \ports(a_{uv})} \pi_{p} &\; = \; \sum_{p \in \ports(a_{vu})} \pi_{p} & & \forall a_{uv} \in A \\
    & & \pi_{p} &\; \in \lbrace 0 ,1 \rbrace \;  & &\forall p \in \bigcup_{a \in A} \ports(a) 
    \end{flalign}

\medskip
\caption{\ac{SR}-based \ac{LC-MCFS} formulated as an \ac{ILP}.}
\label{problem:Green_SR_LP}
\end{problem*}

\begin{theorem}
    For every linecard size~$k \in \N^+$, \ac{LC-MCFS} is NP-hard, already on
    \acp{DAG}, both when ports are assigned to
    predetermined linecards and when they can be assigned to arbitrary ones.
\end{theorem}
\begin{proof}
    We give a Cook reduction from the NP-hard
    Set Cover problem~\cite{DBLP:books/fm/GareyJ79} to \ac{LC-MCFS} on \acp{DAG}.
    In the Set Cover problem, we are given a universe $\univ$ as well as a
    family of sets $\sets = \{S_i \subseteq \univ \ssep i = 1,\dots,\sigma\}$
    and asked for a subfamily~$\setcover \subseteq \sets$ with minimum
    size~$|\setcover|$ such that $\Union_{S \in \setcover} S = \univ$. Given
    such a Set Cover instance~$(\univ, \sets)$, we create an \ac{LC-MCFS}
    instance~$(G,\cp,\demand,k)$ by first constructing a directed
    graph~$G=(V,A)$ (see \Cref{fig:set_cover_to_lcmcf_reduction} for a
    visualization of~$G$):
    Its vertices $V = V_\univ \union V_\sets \union \{z\}$ are made up of two
    vertex sets
    $V_\univ \be \{v_{u} \ssep \forall  u \in \univ\}$ and
    $V_\sets \be \{v^0_S,~\dots,~v^{q}_S \ssep \forall S \in \sets\}$ with
    $q > 2 \cdot (|\sets| + \sum_{S \in \sets} |S|)$
    as well as an additional distinguished vertex~$z$.
    Then, we create a long path of arcs~$P_S \be
    \{(v^0_S,v^1_S),~\dots,~(v^{q-1}_S,v^{q}_S),~(v^{q}_S,z)\}$
    for each set~$S \in \sets$, and let $A = A_\univ \union A_\sets$ where
    $A_\univ \be \{(v_u,v^0_S) \ssep \forall S \in \sets, u \in S\}$ and
    $A_\sets \be \Union_{S \in \sets} P_S$.
    Arc capacities and traffic demands are chosen as follows:
    \begin{align*}
        \cp(a) &\be
            \begin{cases}
                1 &\mbox{if } a \in A_\univ\\
                |U| &\mbox{if } a \in A_\sets,\\
            \end{cases}\\
        \demand_{xy} &\be
            \begin{cases}
                1 &\mbox{if } (x,y) \in V_\univ \times \{z\}\\
                0 &\mbox{otherwise.}\\
            \end{cases}
    \end{align*}

    It remains to show that an optimal solution for this \ac{LC-MCFS} instance
    grants us an optimal solution for the original Set Cover instance.
    A feasible \ac{LC-MCFS} solution $A'$ contains at least one path from
    $v_u$ to $z$ for each item~$u \in \univ$.
    This path consists of one arc from $A_\univ$ and a path $P_S \subseteq
    A_\sets$ with $u \in S$.
    Thus, $\{S \ssep P_S \subset A'\}$ is a feasible solution for the
    original Set Cover instance, i.e., one that contains at least one covering
    set for each item.
    An optimal \ac{LC-MCFS} solution contains each path $P_S$ either fully or
    not at all since a partial containment would only worsen the value of the
    solution without changing its feasibility.
    Whenever a path~$P_S$ is taken into the solution, this increases the number
    of linecards for every vertex $v^i_S$ with $i \in \{1,\dots,q\}$, $S \in
    \sets$ by one (if $k\geq2$; by two if $k = 1$).
    This also holds in the case of predetermined linecards as there is
    always only one way to assign the two ports to the linecards on these vertices.
    Since the overall number of possible ports on other vertices
    $\sum_{v \in V_U \union \{v^0_S | S \in \sets\} \union \{z\}} \deg(v) = 2
    \cdot (|\sets| + \sum_{S \in \sets} |S|)$ is less than $q$, the number of
    linecards belonging to paths~$P_S \subseteq A'$ dominates the objective
    value of the solution.
    An optimal \ac{LC-MCFS} solution, thus, always contains a minimum number of
    paths~$P_S$, and the Set Cover solution $\{S \ssep P_S \subset A'\}$
    is optimal.
\end{proof}

To adapt the reduction to the setting where any arc has an opposing directed arc
with the same capacity, we can introduce additional traffic demands such that
the resulting traffic necessarily blocks the back arcs:
\begin{align*}
    \demand_{xy} &=
        \begin{cases}
            1 &\mbox{if } (x,y) \in V_\sets \times V_\univ\\
            |U| &\mbox{if } (x,y) \in \{z\} \times V_\sets.\\
        \end{cases}
\end{align*}

\section{SR-based Heuristic}\label{sec:greenSR}
In this section, we describe an \ac{SR}-based heuristic for \ac{LC-MCFS} that we developed since computing an optimal \ac{LC-MCFS} solution is often too slow on real-world networks.
The heuristic is based on the \ac{ILP} introduced in \cite{bhatia}; its core idea is to decrease the number of possible paths between two distinct nodes.
Instead of allowing every possible path, we only use paths that are made up of two node segments, i.e., those that can be created by the concatenation of two shortest paths. To determine the amount of traffic that is routed over arc~$a$, we need to define two functions:
The first function~$F_{uw}(a)$ determines the amount of traffic routed over arc~$a$ when routing a unit flow via the shortest path between~$u$ and~$w$. The second function
\begin{align*}
    g_{uv}^{w}(a) = \demand_{uv} \cdot (F_{uw}(a) + F_{wv}(a))
\end{align*}
determines the amount of $u$-$v$-traffic routed over arc $a$ when choosing the intermediate node $w$.
The overall traffic on arc $a$ is given by
\begin{align*}
    \sum_{(u,v) \in V^{2}} \sum_{w \in V} g_{uv}^{w}(a)x_{uv}^{w}
\end{align*}
with $x_{uv}^w$ denoting the fraction of $u$-$v$-traffic that is routed via~$w$.
Now, we can formulate a 2-\ac{SR} version of \ac{LC-MCFS}, see \Cref{problem:Green_SR_LP}.
We keep the port variables and the target function from \Cref{problem:ideal_lp} as we still want to minimize the number of active linecards in the network.
However, the constraints have changed:
The first two types of constraints ensure that all traffic is routed via intermediate nodes and that all routed traffic is positive.
Similar to \Cref{problem:ideal_lp}, the third type of constraints ensures that the \ac{MLU} is kept below $\theta$, and the fourth one models full duplex connections.

\section{Evaluation}\label{sec:evaluation}
In this section, we present the datasets and setup used for our evaluation.
This is followed by a presentation and discussion of the results.

\begin{table}
	\centering
	\caption{Overview over the selected Repetita instances.}
	\label{tab:repetita_instances}
	\scriptsize
	\begin{tabular}{c l r r}
		\toprule
		Identifier & Name             & Nodes & Edges \\
        \midrule
		A          & GridNet          & 9     & 40    \\
		B          & FCCN             & 23    & 54    \\
		C          & BandCon          & 21    & 56    \\
		D          & BTAsiaPacific    & 20    & 62    \\
		E          & FuNet            & 26    & 62    \\
		F          & RedIris          & 19    & 64    \\
		G          & GtsCzechRepublic & 32    & 66    \\
		H          & BTEurope         & 24    & 74    \\
		I          & Renater2008      & 33    & 86    \\
		J          & JanetBackbone    & 29    & 90    \\
		K          & DeutscheTelekom  & 30    & 110   \\
		L          & Renater2010      & 43    & 112   \\
		M          & Geant2012        & 40    & 122   \\
		N          & Forthnet         & 62    & 124   \\
		O          & Chinanet         & 42    & 132   \\
		P          & Ulaknet          & 82    & 164   \\
		Q          & Garr2012         & 61    & 178   \\
		R          & RedBestel        & 84    & 202   \\
		S          & Uninett2010      & 74    & 202   \\
		T          & Oteglobe         & 83    & 204   \\
		U          & Globenet         & 67    & 226   \\
		\bottomrule
	\end{tabular}
\end{table}
\subsection{Datasets}
We used two datasets for the evaluation, the first one consisting of 21 randomly selected instances from the Repetita dataset~\cite{repetita} (cf.~\Cref{tab:repetita_instances}).
This dataset is also included in the Topology Zoo~\cite{topo-zoo}, which contains several real-world backbone network topologies with artificial traffic data. The Topology Zoo is a publicly available dataset that was created to test and optimize different \ac{TE} approaches.
Thus, the traffic data is designed to be challenging in terms of the \ac{MLU}.
Even with an \ac{MLU}-optimal routing algorithm, the \ac{MLU} is at least $90$\%.
Hence, we scale down the traffic to obtain some unused capacity, choosing the scaling factor $0.5$ to mimic the effect of low utilization.

Furthermore, the topologies do not contain information about the employed hardware.
Thus, we assume that every link consists of $4$ parallel links.
Each of the parallel links corresponds to one port per connected router.
Further, for every $8$ ports of a router, we assume that there is a corresponding linecard providing service exclusively in the backbone network.
To give an example: When a router in the given topology supports up to $15$ links, we assume that only one of its linecards is used exclusively in the network. The other linecard may also provide a sixteenth link, which is not mentioned in the topology data, to connect the backbone to its customers. Thus, we assume that it cannot be deactivated. For example, if the A99-8X100GE-TR is used within the network, these assumptions are fulfilled.
With these assumptions it is possible to evaluate our approach on many different topologies, enabling us to gain better insight into the potential of our approach.

Our second dataset consists of real traffic and topology data collected from the backbone network of a globally operating \ac{ISP} in 2020 and 2022. The dataset includes daily topology snapshots; the traffic matrices were measured every $15$ minutes with the method described in \cite{methode-horneffer}.
As the topology evolved over time, the number of nodes varies from $160$ to $190$, while the number of links varies from $3700$ to $4600$. We made the same assumptions about this topology as for the Topology Zoo data, i.e., that there is a linecard providing service exclusively in the network for every $8$ ports of a router, and that each of these ports corresponds to one link.

\begin{figure}
		\centering
		\includegraphics[width=.98\linewidth]{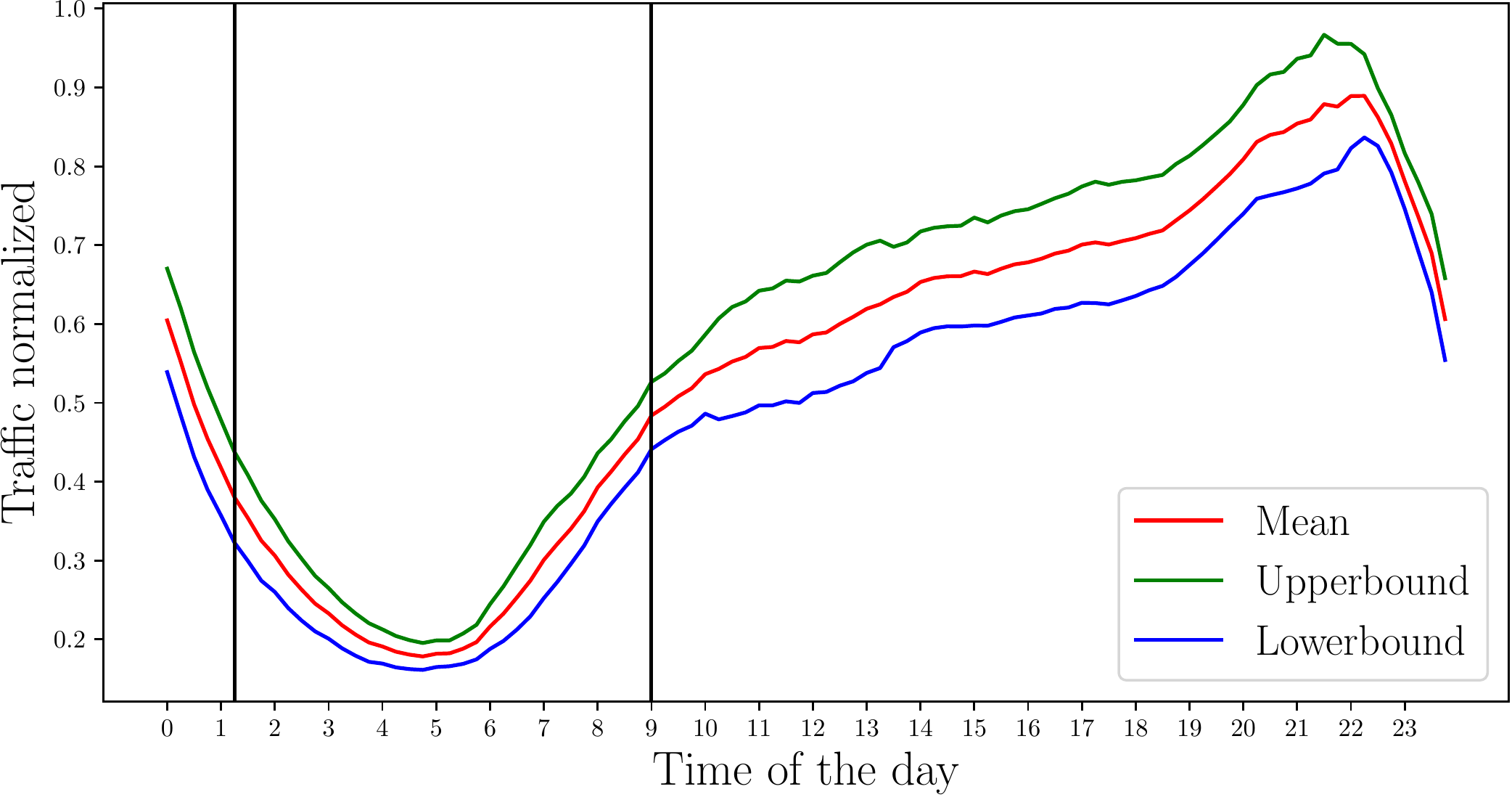}

	\caption{Traffic volume throughout the day in 2022.}
	\label{FIG:traffic_22}
\end{figure}

As the number of instances is too large to evaluate our approach on every instance, we choose a random instance for every month in 2020 and 2022. To find times of the day when the traffic is low enough to deactivate some linecards, we used a probabilistic model where the amount of traffic is represented by a time-dependent normally distributed random variable. For every time of the day, we calculated the expected mean and deviation of the traffic, yielding a 70\% confidence interval. Thus, for each given time point, the traffic is with a certainty of 70\% inside this interval. Moreover, we can assure with a certainty of 85\% that the amount of traffic is lower than the upper bound. We visualized the results of this procedure in Figure \ref{FIG:traffic_22}. The red line marks the upper bound while the green line marks the lower bound of the confidence interval. In accordance with these findings, we evaluate our approach on traffic matrices captured at 01:15: At this time of day, the amount of traffic is lower than 50\% of the daily maximum, but above 40\%. Thus, we can expect that a reduced topology that is capable of managing this amount of traffic can be used until 09:00.

\subsection{Algorithms and Resource Demands}
We compare two algorithms, one that yields an optimal solution for the \ac{LC-MCFS} problem and our 2-\ac{SR}-based heuristic. The optimal solution of the \ac{LC-MCFS} problem is obtained by solving the \ac{ILP} stated in \Cref{problem:ideal_lp} with CPLEX~\cite{cplex}. Similarly, CPLEX was also used to implement the \ac{ILP} for our heuristic as stated in \Cref{problem:Green_SR_LP}. The execution of either algorithm required several hundred Gigabytes of RAM. Therefore, we used a computer equipped with two AMD EPYC 7452 CPUs, 256GB of RAM and 64-bit Ubuntu~20.04.1.

In real-life practical scenarios, networks need to be able to handle sudden traffic spikes.
To simulate this need for extra capacity, we choose an upper bound~$\theta$ of 70\% for the \ac{MLU}. This upper bound ensures that a traffic increase of 40\% can still be handled without congestion.

\subsection{Results}
\begin{figure*}
		\centering
		\begin{subfigure}{.49\textwidth}
		\centering
		\includegraphics[width=.98\linewidth]{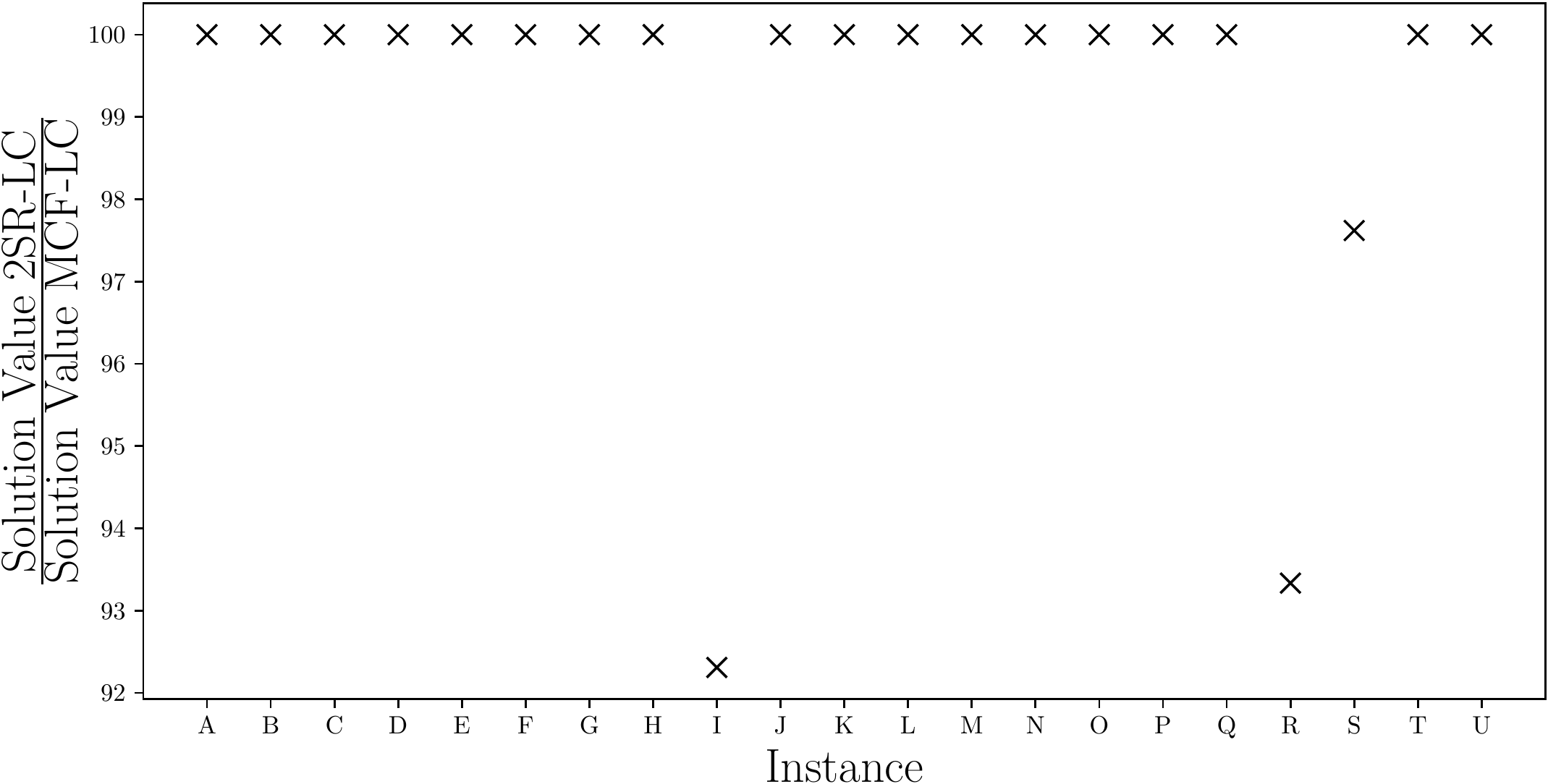}
		\caption{2SR and MCF compared}
		\label{FIG:MCFvs2sr}
\end{subfigure}
\begin{subfigure}{.49\textwidth}
		\centering
		\includegraphics[width=.98\linewidth]{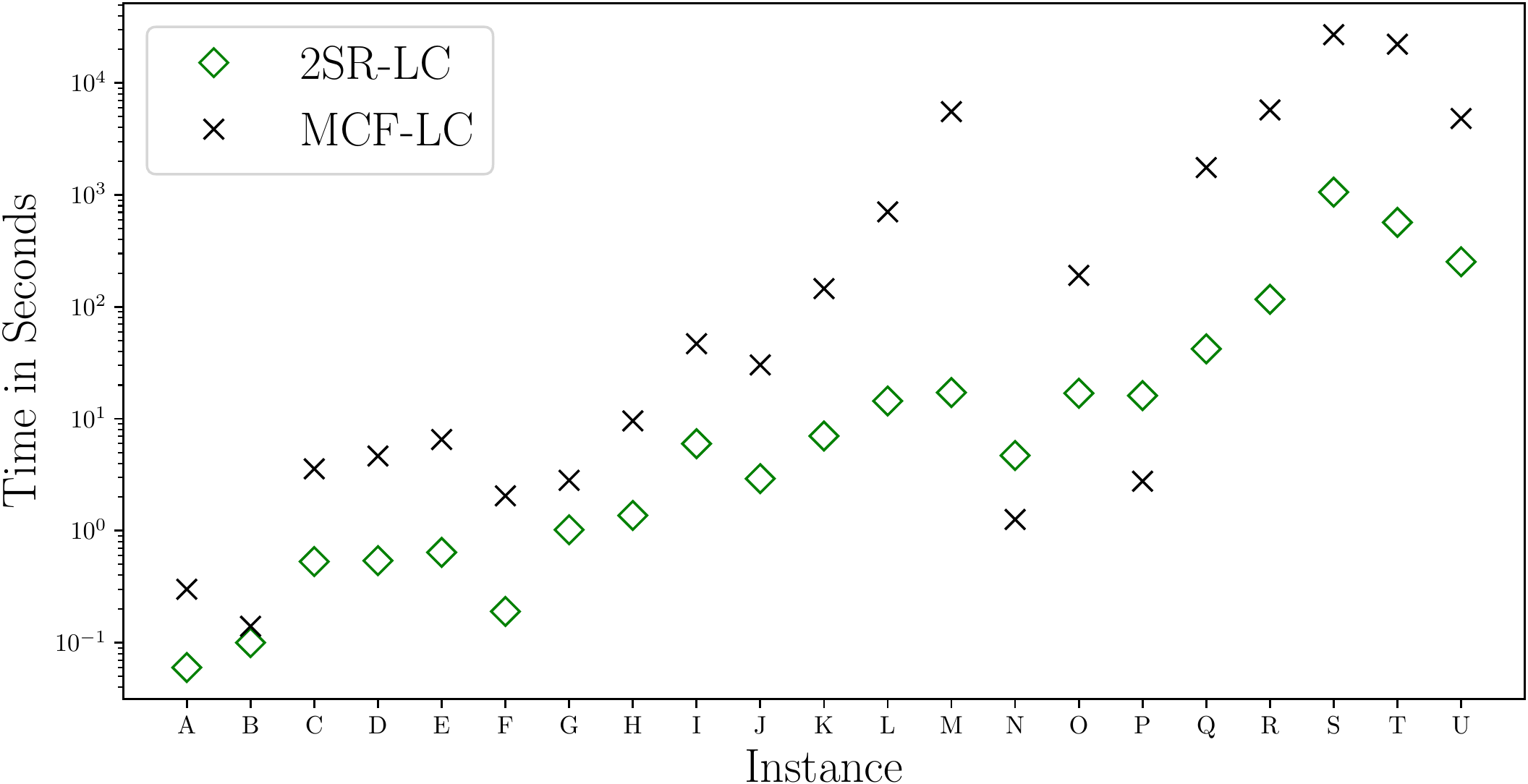}
		\caption{Computation Time}
		\label{FIG:Time}
\end{subfigure}
\caption{Results Repetita Data}
\end{figure*}
\begin{figure*}
\centering
\begin{subfigure}{.49\textwidth}
		\centering
		\includegraphics[width=.98\linewidth]{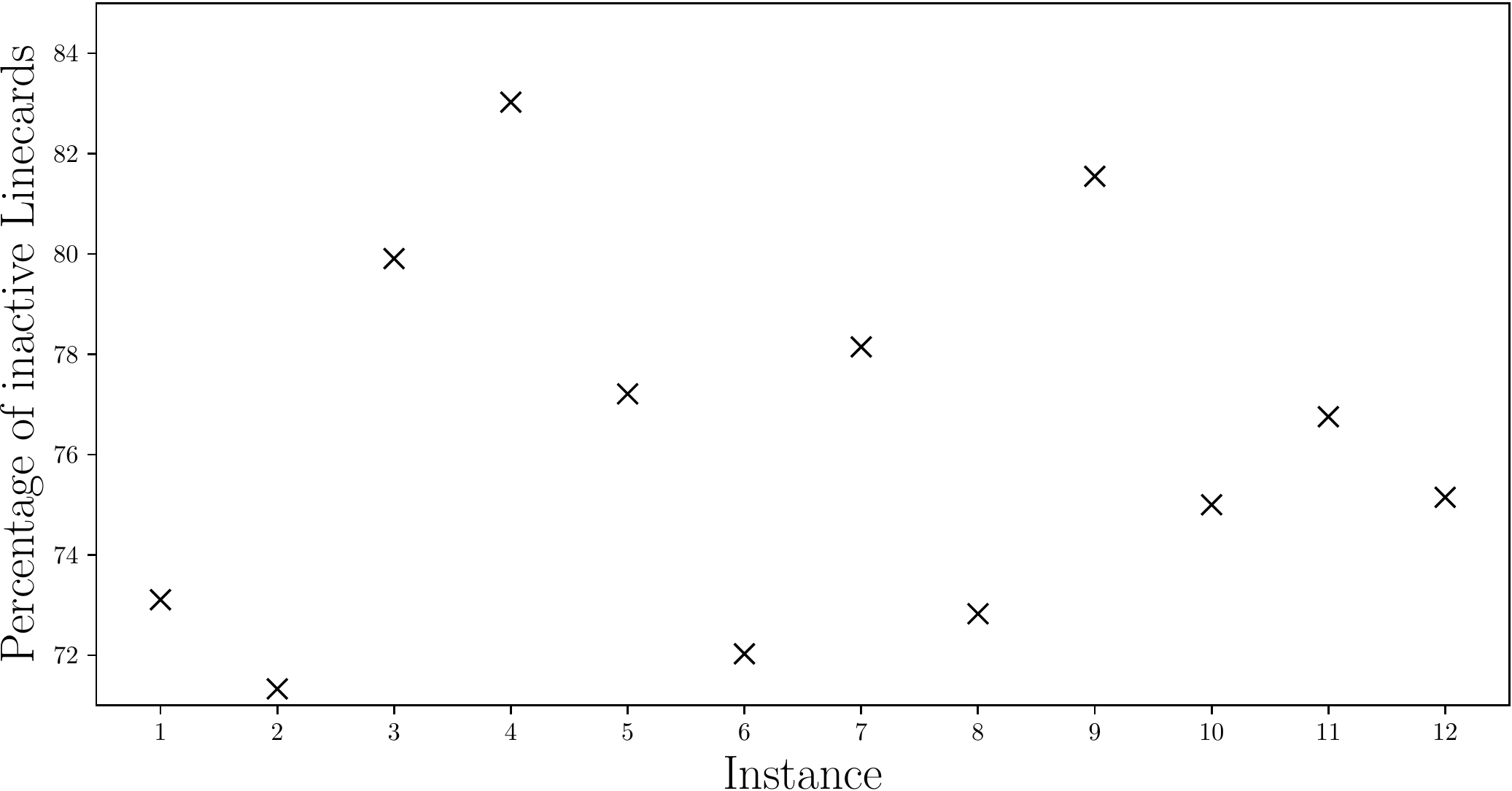}
		\caption{Results ISP Data 2020}
		\label{FIG:real_20}
\end{subfigure}
\begin{subfigure}{.49\textwidth}
		\centering
		\includegraphics[width=.98\linewidth]{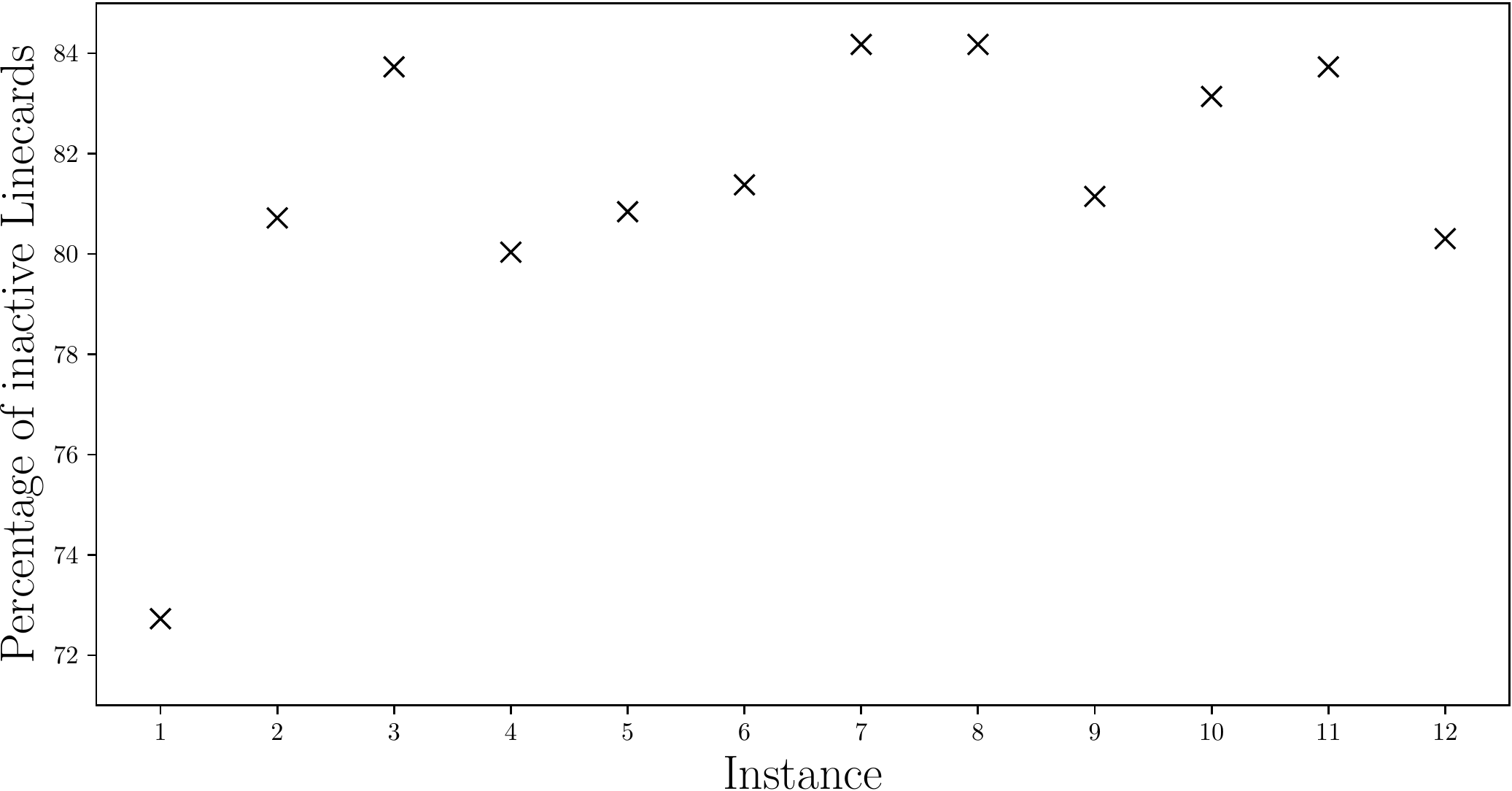}
		\caption{Results ISP Data 2022}
		\label{FIG:real_22}
\end{subfigure}
\caption{Results ISP Data}
\label{FIG:real_topologies}
\end{figure*}

In this section, we present and discuss the results of our evaluation, beginning with results pertaining to the Repetita dataset. We compare 2SR-LC, our newly developed heuristic, with the \ac{LC-MCFS} \ac{ILP}, which we refer to as MCF-LC. As MCF-LC is an upper bound for the number of inactive linecards, this lets us gain insights into the optimality of our heuristic.
\Cref{FIG:MCFvs2sr} shows a scatter plot of solution values of 2SR-LC divided by the optimal MCF-LC solution values.

For $18$ out of $21$ instances, 2SR-LC yields an optimal solution. However, for the instances I~(Renater08), R~(Red-Bestel), and S~(Uninett2010), the heuristic did not manage to obtain an optimal solution.
Nevertheless, 2SR-LC always manages to find a solution that is at most ten percent worse than the optimal solution.
In absolute numbers, 2SR-LC differed only by one linecard for all three instances.

In \Cref{FIG:Time}, we compare the computation time of both approaches. The topologies are sorted by their number of edges. As one would expect, the running time in general seems to increase with the number of edges contained in the topology; this holds true for both approaches.
Further, it becomes evident that 2SR-LC is faster than MCF-LC. There are two exceptions to this: the instances Forthnet and Ulaknet. Note that both algorithms do not need longer than ten seconds to solve these instances.
Since they are structured like trees, the algorithms have little possibility to steer the traffic in these networks. In general, it holds that the larger the instance in terms of edges, the larger the computation time difference.
Here, the benefit of our heuristic becomes very clear: In some cases, 2SR-LC is almost 100 times faster than MCF-LC.

For all 21 instances, the resulting \ac{MLU} was in a range from 0.68 to 0.7. 2SR-LC and MCF-LC both tend to make full use of the capacity threshold. Thus, it is worth to further investigate which upper bound serves as the best tradeoff between reliability in cases of traffic spikes and potential for switching off hardware.

In \Cref{FIG:real_topologies}, we have visualized our results based on the topologies and traffic measured in the backbone network of a globally operating \ac{ISP}. It was not possible to calculate an MCF-LC solution for these instances. As they consist of far more nodes and edges than the Repetita topologies, the computation took too long to obtain even a feasible solution. In fact, the CPLEX solver did not manage to find at least a feasible solution after one week of running time. Nevertheless, we decided to calculate a routing policy with 2SR-LC.
This gives us some insight into the potential of what could be achieved when additional real-world constraints were to be considered in the model.

The fraction of inactive linecards regarding the 2020 topology can be found in \Cref{FIG:real_20}. It varies from at least $72$\% to $83$\% at most. The instances from 2022 produced similar values, with at least $73$\% and at most $84$\% of all linecards being switched off. The \ac{MLU} is $0.7$ for every instance of 2020 and 2022.

With this evaluation, we have shown that our approach is able to handle instances with more than 400 nodes.
Moreover, it is clearly possible to obtain a near-optimal solution of the MCF-LC problem by shrinking the solution space to all paths that can be described as 2SR paths. As we evaluated our approach on 21 randomly chosen datasets representing different backbone topologies, we suggest that especially for topologies occurring in backbone networks, our heuristic is a viable way to find close-to-optimal solutions. In addition, since we decided to use 2-\ac{SR}, which is a very lightweight \ac{TE} technology, the approach can be implemented without adding unnecessary overhead to the network.

More generally, we have shown that there is a huge amount of saving potential within the topology of a backbone network. Nevertheless, to be applicable to a backbone network in a practical setting, we have to consider several other constraints. For example, we did not consider any resilience constraints beyond an upper bound for the \ac{MLU}. In the future, we have to consider that link failures may happen, and take care of this problem.

\section{Conclusion and Future Work}
Within this paper, we proposed a new green \ac{TE} approach. Based on a closer inspection of the energy consumption in a backbone network, this approach aims to minimize the number of active linecards in the network as they consume the vast majority of the power.

We analyzed the overall complexity of this idea and showed that this problem is NP-hard. In fact, we could show that this even holds for simple directed acyclic graphs.

Furthermore, we developed an \ac{SR}-based heuristic that drastically decreases the number of possible packet paths by using just two segments. Contrary to older \ac{TE} technologies, this keeps the overhead induced in the network small.
During our evaluation, we showed that for most instances, the heuristic yields nearly optimal results while also significantly reducing the computation time. We further evaluated our approach on real-world topologies, showing the potential of our approach in terms of energy savings.

In the future, more constraints have to be considered. In particular, the resiliency of the reduced topologies has to be examined. Nevertheless, we have shown that it is possible to use 2-SR for green \ac{TE}. This approach reduces the algorithmic complexity of the problem such that it is possible to find a good solution in a reasonable amount of time.

\section*{Acknowledgment}
This work was supported in part by the German Research Foundation (DFG), Project No. AS 341/7-1.

\bibliographystyle{IEEEtran}
\bibliography{main}
\balance

\end{document}